\documentclass[preprint,12pt]{elsarticle}

\usepackage{amssymb}
\usepackage{multirow}

\newtheorem{lem}{Lemma}
\newenvironment{proof}[1][Proof]{\noindent\emph{#1.} }{\ $\Box$\smallskip}

\journal{Computers \& Operations Research}

\begin{document}

\begin{frontmatter}



\title{Tactical Fixed Job Scheduling \\ with Spread-Time Constraints}


\author[label1,label2]{Shuyu Zhou}
\author[label3]{Xiandong Zhang}
\author[label4]{Bo Chen}
\author[label2]{Steef van de Velde}

\address[label1]{Department of Mathematics, East China University of Science and Technology, \\ Shanghai 200237, China}
\address[label2]{Rotterdam School of Management, Erasmus University, \\ 3000 DR Rotterdam, The Netherlands}
\address[label3]{School of Management, Fudan University, \\ Shanghai 200433, China}
\address[label4]{Warwick Business School, University of Warwick, \\ Coventry CV4 7AL, UK}

\begin{abstract}
We address the tactical fixed job scheduling problem with spread-time constraints. In such a problem, there are a fixed number of classes of machines and a fixed number of groups of jobs. Jobs of the same group can only be processed by machines of a given set of classes. All jobs have their fixed start and end times. Each machine is associated with a cost according to its machine class. Machines have spread-time constraints, with which each machine is only available for $L$ consecutive time units from the start time of the earliest job assigned to it. The objective is to minimize the total cost of the machines used to process all the jobs. For this strongly NP-hard problem, we develop a branch-and-price algorithm, which solves instances with up to $300$ jobs, as compared with CPLEX, which cannot solve instances of $100$ jobs. We further investigate the influence of machine flexibility by computational experiments. Our results show that limited machine flexibility is sufficient in most situations.
\end{abstract}


\end{frontmatter}

% \linenumbers


\section{Introduction}
The \emph{basic} fixed job scheduling (FJS) problem, also known as interval scheduling problem or $k$-track assignment problem, was firstly described by \citet{dantzig1954minimizing}. In such a problem, there are a set ${\cal J}=\{J_1,\ldots,J_n\}$ of $n$ jobs and each job $J_j$ ($1\le j\le n$) requires processing without interruption from a given start time $r_j$ to a given end time $d_j$. Each machine, while available all the time, can process at most one job at a time. The objective is to determine the minimum number of identical machines needed to process all the jobs. \citet{gupta1979optimal} give an $O(n \log n)$ time exact algorithm and they show that such a running time is the best possible.

In many production environments, jobs are classified into disjoint groups and machines into different classes. A machine of a specific class can only process jobs of some specific groups. In particular, a single-purpose machine is inflexible and can only process jobs of a specific group. Added with these constraints, the FJS problem becomes difficult. There are two main variants of such a general FJS problem: operational and tactical.

In the operational FJS problem, a fixed number of identical machines are given. One is to select a subset of jobs from ${\cal J}$ to process so as to maximize the total profit from processing these jobs, where a profit is generated from processing a job. The problem has its roots on the capacity planning of aircraft maintenance personnel for KLM Royal Dutch Airlines (\citet{kroon1995exact}, \citet{Kroon1990job}), where each arriving plane requires one or more maintenance jobs, each within a fixed time period, and each available engineer is licensed to carry out jobs on at most two different aircraft types. Each maintenance job has a priority index. The objective is to find an assignment of aircraft to engineers so as to maximize the total priority index of all assigned maintenance jobs. \citet{kroon1995exact} provide an exact algorithm when there is only one single machine class and two approximation algorithms when there are multiple machine classes. \citet{eliiyi2006spread}, \citet{bekki2008operational}, \citet{eliiyi2009fixed}, \citet{solyali2009operational}, and \citet{eliiyi2010working} study the problem with various job characteristics and machine environments, such as spread-time constraints (see below for more details), working-time constraints, machine-dependent job weights, uniform parallel machines, etc.

The tactical FJS problem is a dual to its operational counterpart stated above, in which one is to minimize the total cost of processing all jobs of ${\cal J}$, where usage of each machine incurs a cost. This problem arose from the gate capacity planning at Amsterdam Schiphol Airport (\citet{kroon1997exact}). \citet{kroon1997exact} modeled each incoming aircraft as a fixed job with given start and end time, and modeled each gate as a machine. Gates could handle aircraft only from a predetermined set of aircraft types. The problem is to serve all incoming aircraft on time with minimum total number of gates.  If the number $c$ of machine classes is fixed and $c>2$, \citet{kolen1992license} prove that the problem is strongly NP-hard. On the other hand, if $c$ is not fixed, \citet{kroon1997exact} prove that the problem is strongly NP-hard even if preemption is allowed. They also present an exact branch-and-bound algorithm for solving the problem to optimality.

For a general overview of research in interval scheduling, we refer the reader to \citet{kovalyov2007fixed} and \citet{kolen2007interval}.

In many practical applications of the FJS model, the \emph{spread-time constraints} are also important, with which each machine is available only for $L$ consecutive time units from the start of the earliest job assigned to it. Such constraints first appeared in the bus driver scheduling problem studied by \citet{martello1986heuristic}. \citet{fischetti1987fixed} show that the basic FJS problem with spread-time constraints is NP-hard and they develop an exact branch-and-bound algorithm. For the same problem, \citet{fischetti1992approximation} provide a 2-approximation algorithm that runs in $O(n\log n)$ time.

In this paper, we address the tactical (general) FJS problem with spread-time constraints (TFJSS). To the best of our knowledge, this study is the first to address the problem. Since the problem is strongly NP-hard, we provide a branch-and-price algorithm, which  solves to optimality randomly generated instances with up to $300$ jobs within one hour.  Instances with 100 jobs are well solved to optimality within 40 seconds. This is in contrast to the fact that, on the standard ILP formulation of the problem, CPLEX cannot solve instances with $100$ jobs within one hour. With the same algorithm, we solve the TFJSS problem with additional constraints.

Branch-and-price algorithms have been proved to be a successful technique in solving large mixed integer programs (\citet{barnhart1998branch}). In recent years, many problems in diversified fields are solved with this kind of approach. These problems include, for example, the $p$-median location problem (\citet{senne2005branch}), the batch processing machine scheduling problem (\citet{rafiee2010branch}), and the surgical case sequencing problem (\citet{cardoen2009sequencing}), etc. Given an optimization problem, the establishment of a branch-and-price approach includes the designs of a pricing algorithm and a branching strategy as well as an optimality proof. The efficiency of a branch-and-price algorithm relies heavily on those designs. Our pricing algorithm is based on a dynamic programming model, which is a familiar approach in many branch-and-price algorithms. Our branching strategy is based on precedence relations. \citet{belien2006scheduling} use the same branching idea in solving the trainees scheduling problem in a hospital department. Our computational experiments show that our design is efficient.

The paper is organized as follows. After providing two ILP formulations for the TFJSS problem in Section~\ref{sec:model_algo}, one standard and the other for column generation purpose, we describe our branch-and-price algorithm in Section~\ref{se:process}. Computational results of the algorithm are provided in Section~\ref{sec:computation}, in which we test the computational effectiveness of our algorithm and, additionally, test a theory on machine flexibility as stated in \citet{jordan1995principles}.  We draw some conclusions in Section~\ref{sec:conclusions}. In the appendix, we provide a neighborhood search algorithm that can be embedded in our branch-and-price algorithm to improve its efficiency in many cases.

\section{Formulations}
\label{sec:model_algo}

%Let $p_j$ be the processing time of job $J_j$ (i.e., $p_j = d_j - r_j$) for $j=1, \ldots, n$.

Denote by $M^i$ the set of machines of class $i$ ($1\le i\le c$) and by $C_j$ ($1\le j\le n$) the set of machine classes containing machines that can process job $J_j$. Let $w_i$ and ${\cal J}^i$ be, respectively, the cost of any machine of $M^i$ and the set of jobs that can be processed by a machine of $M^i$.

%Denote $m=m_1+\cdots+m_c$.

\subsection{Basic formulation}

Without loss of generality, we first re-index all jobs so that $r_1\leq r_2\leq \cdots \le r_n$. A pair of jobs $J_j, J_k \in {\cal J}^i$ with $j<k$ are said to be \emph{compatible} with each other if they can be assigned together and processed by a single machine of $M^i$, i.e., $r_k \ge d_j$ and $d_k-r_j \le L$. For any $J_j\in {\cal J}^i$ ($1\le i\le c$), let
\[
\mathcal{A}_j=\{J_k \in \mathcal{J}^i: k>j \textrm{ and job $J_k$ is not compatible with } J_j\}.
\]
That is, ${\cal A}_j$ consists of all jobs having larger start times than $r_j$ that cannot be assigned to a machine together with job $J_j$ due to either overlapping processing intervals or the spread-time constraint.

Define binary decision variable $y_k^i$ to represent whether or not machine $k \in M^i$ is used, and binary decision variable $x_{jk}^i$ to denote whether or not job $J_j \in \mathcal{J}^i$ is assigned to machine $k \in M^i$. The TFJSS problem can be formulated as an ILP as follows:

\begin{equation}\label{eq:ZIP_obj}
z_{0}=\min \sum_{i=1}^{c} w_i \sum_{k\in M^i} y_k^i
\end{equation}
subject to
\begin{eqnarray}
x_{jk}^i \leq y_k^i, & & (J_j\in \mathcal{J}^i, k \in M^i, i=1,\dots, c) \label{eq:x_ijk<=y_ik} \\
x_{jk}^i+x_{\ell k}^i \leq 1, & & (J_\ell\in {\cal A}_j, J_j\in \mathcal{J}^i, k\in M^i, i=1,\dots,c) \label{eq:x_ijk+x_ilk<=1} \\
\sum_{i\in C_j}\sum_{k\in M^i} x_{jk}^i = 1, & & (j=1,\dots,n) \label{eq:sum_x_ijk=1}\\
x_{jk}^i, y_k^i \in \{0,1\}. & & (J_j \in \mathcal{J}^i, k\in M^i, i=1,\dots, c) \label{eq:x_ijk_in0,1}
\end{eqnarray}
The objective function (\ref{eq:ZIP_obj}) is to  minimize the total cost of used machines. Constraints (\ref{eq:x_ijk<=y_ik}) guarantee that machine $k$ is used when job $J_j$ is assigned to that machine. Constraints (\ref{eq:x_ijk+x_ilk<=1}) make sure that the jobs are compatible on any used machine. Constraints (\ref{eq:sum_x_ijk=1}) state that each job is processed exactly once. Finally, (\ref{eq:x_ijk_in0,1}) are the binary constraints on the assignment variables.

\subsection{Formulation for column generation}

Define a \emph{single-machine schedule} as a string of jobs that are compatible with each other. For any $i=1,\ldots, c$, let $a_{js}^i$ be a binary constant that is equal to 1 if and only if job $J_j\in \mathcal{J}^i$ is a job in single-machine schedule $s$. Accordingly, column $a_s^i=(a_{1s}^i, \dots, a_{ns}^i)^T$ represents the jobs in single-machine schedule $s$ that can be processed by a machine of $M^i$. For convenience, in what follows, we often refer a single-machine schedule also as a \emph{column}.

Let $S^i$ be the set of all single-machine schedules of jobs in $\mathcal{J}^i$. We introduce binary decision variables $x_s^i$ ($s \in S^i$ and $i=1,\dots,c$) such that $x_s^i=1$ if and only if single-machine schedule $s$ is used for processing the corresponding jobs on a machine of $M^i$. Therefore, our problem is to select a set of single-machine schedules so that all jobs in $\mathcal{J}$ are processed exactly once and the total machine cost is minimized. Consequently, the problem can be formulated in a form of column generation as follows:
\begin{equation}\label{eq:CG_obj}
z_{I}=\min \sum_{i=1}^c\sum_{s\in S^i} w_i x_s^i
\end{equation}
subject to
\begin{eqnarray}
\sum_{i=1}^c\sum_{s\in S^i} a_{js}^ix_s^i\geq 1, && j=1,\dots,n, \label{eq:sum_ajsxs=1} \\
x_s^i\in \{0,1\} && s\in S^i, i=1,\dots,c. \label{eq:xs_in0,1}
\end{eqnarray}
Constraints (\ref{eq:sum_ajsxs=1}) ensure that each job is executed at least once. The reason why we use inequality rather than equality is to constrain the dual space and speed up the convergence. Note that the dual variables corresponding to constraint (\ref{eq:sum_ajsxs=1}) are nonnegative. The binary constraints (\ref{eq:xs_in0,1}) ensure machine schedule $s$ is selected once or not at all.

Since the number of columns involved in the above formulation can go exponential large with increased $n$, we apply column generation method. To this end, we relax the binary constraints on the variables to obtain the LP relaxation of ILP~(\ref{eq:CG_obj})--(\ref{eq:xs_in0,1}). Denote by $X_L$ an optimal solution of the LP relaxation and by $z_{L}$ the corresponding optimal objective value.

\section{The solution process}
\label{se:process}

The column generation method solves the LP relaxation problem in which only a subset of the variables are available. An initial solution is generated by a greedy algorithm presented in Section~\ref{sct:Init}. New columns will be added, which may decrease the solution value, if the optimal solution has not been determined yet. These new columns are identified through a pricing algorithm that solves an optimization problem at each iteration. We present the pricing algorithm in Section~\ref{sct: pricing}. To ensure that an integer solution is found, we apply a branch-and-price approach together with column generation. We present the details of the whole algorithm in Section~\ref{sct: B&P}.

\subsection{Selection for initial solution}\label{sct:Init}

Our solution process starts with the following algorithm, which selects an initial solution.

\bigskip\noindent
\textbf{Algorithm GR}
\begin{description}
\vspace{-6pt}
\item[Step 1.]Re-index all jobs if necessary so that $r_1 \le r_2 \le \cdots \le r_n$. Let $S = \emptyset$ and $j := 1$.
\vspace{-6pt}
\item[Step 2.]If $j > n$, then stop and output a feasible LP solution $X$ that corresponds to the set $S$ of single-machine schedules (columns).
\vspace{-6pt}
\item[Step 3.]Check whether there exists a single-machine schedule $s \in S$ such that job $J_j$ can be appended to $s$ to form a valid single-machine schedule $s'$. (a) If yes, set $S:= S\backslash\{s\}\cup\{s'\}$; (b) if no, let $\ell = \min\{i: i\in C_j\}$ and construct a new single-machine schedule $s^\ell$, which consists of the single job $J_j$, and set $S := S \cup \{s^\ell\}$; (c) in either case, set $j:= j + 1$ and go to Step 2.
\end{description}

\subsection{The pricing algorithm for column generation}
\label{sct: pricing}

According to the duality theory, a solution to a minimization LP problem is optimal if and only if the \emph{reduced cost} of each variable is nonnegative. The reduced cost $P_s^i$ of (relaxed) variable $x_s^i$ is given by
\[
P_s^{i}= w_i - \sum_{J_j\in \mathcal{J}^i}\lambda_j a_{js}^i,
\]
where $\lambda_1,\dots,\lambda_n$ are the values of the dual variables corresponding to the current LP solution. To test whether the current solution is optimal, we determine if there exists a single-machine schedule $s\in S^i$ so that $P_s^{i}$ is negative. To this end, we solve the pricing problem of finding the machine schedule in $S^i$ with the minimum reduced cost for each $i$. Because $w_i$ is a constant for any $s\in S^i$, we essentially need to maximize
\begin{equation}\label{eq:pricing_obj}
\tilde{P}_s^i:=\sum_{J_j\in \mathcal{J}^i}\lambda_j a_{js}^i.
\end{equation}
If the current value of $\tilde{P}_s^i  \leq w_i$ for each $i$, we have already found an optimal solution for the LP relaxation problem. Otherwise, we use the following pricing algorithm to search for new columns $\{a_{s}^i\}$ to maximize (\ref{eq:pricing_obj}).

The pricing algorithm is based on dynamic programming (DP) and uses a forward recursion that exploits the property that on each machine the jobs are sequenced in order of increasing indices. Furthermore, since our DP-based pricing algorithm will be combined with a branch-and-bound process and hence used iteratively, it is constrained to generate columns of those single-machine schedules in which each job $J_j$ has a specified set $\mathcal{J}_j^+\subseteq \mathcal{J}$ of immediate successors, where an \emph{immediate successor} of job $J_j$ in a single-machine schedule $s$ is the job that is the first to be processed in $s$ after completion of job $J_j$. Denote job precedence constraints by
\begin{equation}\label{eq:def-P}
    \mathcal{P}=\bigcup_{j=1}^n \mathcal{J}_j^+.
\end{equation}
We refer the following DP-based algorithm as DP($\mathcal{P}$). For $1\le j \le k\le n$ and $1\le i\le c$, let
\[
{F}_i(j,k) = \left\{s\in S^i: J_j, J_k \in \mathcal{J}^i
\textrm{ are the first and last job in } s\right\}
\]
and
\[
f_i(j,k)=\left\{
\begin{array}{ll}
  \max_{s\in {F}_i(j,k)}\tilde{P}_s^i,  &  \textrm{if } {F}_i(j,k)\neq \emptyset \\
  0, &  \textrm{otherwise.}
\end{array}
\right.
\]
Then, after initialization $f_i(j,j)=\lambda_j$ for any $J_j\in {\cal J}^i$ ($i=1, \ldots, c$), the recursion for $k =j+1,\dots, n$ and $J_k\in {\cal J}^i$ is as follows:
\[
f_i(j,k) =   \left \{
\begin{array}{ll}
0, & \mbox{ if } J_k\in {\cal A}_j \\
\max_{\ell:\, j\leq \ell < k;\, J_k \in {\cal J}_{\ell}^+ } f_i(j,\ell) + \lambda_k, & \mbox{ if } J_k\notin {\cal A}_j.
\end{array}
\right.
\]
Let $f_i^*=\max_{1\le j\le k\le n} f_i(j,k)$ ($i=1, \ldots, c$).

If $f_i^*\leq w_i$ for any $i$, then the current LP solution is optimal. Otherwise, it is not and we need to introduce new columns to the problem. Candidates are associated with those $i$ for which $f_i^*> w_i$. An important implementation issue is to determine the number of columns to add to the LP after having solved the pricing algorithm. The more columns we add per iteration, the fewer LPs we need to solve, but the bigger the LPs become. An empirically good choice for our computational platform appears to be adding those ten columns that correspond to those ten values of $i$ for which $f_i^*$ are most positive. The pricing algorithm requires $O(c^{} n^2)$ time and space.

We solve the LP problem after the addition of new columns and calculate the new set of dual variables $\lambda_1, \ldots, \lambda_n$. Then repeat the DP-based pricing algorithm to test the optimality of our new LP solution. Such iteration process is repeated until we find an optimal LP solution to the original problem.

\subsection{The branch-and-price algorithm}\label{sct: B&P}

Recall that $X_L=\{\tilde{x}_s^i:s\in S^i, i=1,\ldots,c\}$ denote an optimal solution of the LP relaxation problem of ILP~(\ref{eq:CG_obj})--(\ref{eq:xs_in0,1}). Throughout this subsection until the formal description of the branch-and-price algorithm at the end, we explain in detail how we deal with $X_L$ at the root node of our branching tree. However, everything is applicable to the case where we are at any descendant node with replacement of $X_L$ by an optimal solution to the corresponding LP problem at that node and ``optimality'' adjusted with respect to the corresponding sub-problem at the node.

If all the components $\{\tilde{x}_s^{i}\}$ of $X_L$ are integers, then we find an optimal IP solution at the node. If not, we apply branch-and-price approach to close the integrality gap. The simple branching strategy of fixing a variable at either zero or one does not work in combination with column generation, since even if we fix the variable at zero, a pricing algorithm may again make this variable positive. Therefore, we use a different branching strategy, which is based on the notion of  predecessors and successors.

For any given single-machine schedule $s$, all the jobs in $s$ are processed in a unique sequence of increasing start times. Therefore, any job in $s$ has a unique immediate \emph{predecessor} and \emph{successor}, where we define the immediate predecessor (respectively, successor) of the first (respectively, last) job in $s$ as a dummy job $J_0$ (respectively, $J_{n+1}$). For any optimal LP solution $X_L$, let
\[
\Phi(X_L)=\{s: s\in S^i \textrm{ such that $0<\tilde{x}_s^{i}<1$ for some $1\le i\le c$}\}.
\]
That is, $\Phi(X_L)$ consists of those single-machine schedules (columns) corresponding to fractional components of $X_L$. Therefore, if $X_L$ is an integral solution, then $\Phi(X_L)=\emptyset$. Denote
\[
  \mathcal{J}_f(X_L) = \{J_j\in \mathcal{J}: \textrm{$J_j$ is a job in $s\in \Phi(X_L)$}\},
\]
and for any $s \in \Phi(X_L)$ and any $J_j\in\mathcal{J}_f(X_L)$,
\begin{eqnarray*}
  S_j(X_L) &=& \{s: s \textrm{ contains job } J_j \}, \\
  \mathcal{I}(s) &=& \{i: 1\le i\le c,\, s\in S^i \textrm{ with } 0<\tilde{x}_s^{i}<1 \}.
\end{eqnarray*}
In other words, jobs in $\mathcal{J}_f(X_L)$ are those that are scheduled only fractionally in $X_L$ (and hence not legitimately scheduled yet) on machines of classes $i\in \mathcal{I}(s)$ with $s \in \Phi(X_L)$, and $S_j(X_L)\subseteq \Phi(X_L)$ are those columns that contain job $J_j$. Observe that, if a job is already in a column $s$ corresponding to an $\tilde{x}_s^{i}=1$, then this job can be removed from any column $s'\in \Phi(X_L)$ without compromising the optimality of $X_L$. Therefore, without loss of generality, unless $\Phi(X_L)=\emptyset$ we have
\begin{equation}\label{eq:feasibility}
 \sum_{s\in S_j(X_L)}\sum_{i\in \mathcal{I}(s)} x_s^i\geq 1\ \textrm{ for any } J_j\in \mathcal{J}_f(X_L).
\end{equation}

Our branching strategy is based on the following two lemmas concerning predecessors and successors.
\begin{lem}\label{le:same_successor}
If each job $J_j\in\mathcal{J}_f(X_L)$ has the same immediate successor in every $s\in S_j(X_L)$, then one can find an optimal LP solution $X_L^*$ in which each job $J_j\in\mathcal{J}_f(X_L^*)$ has the same immediate predecessor and successor in every $s\in S_j(X_L^*)$.
\end{lem}
\begin{proof}
Let job $J_\ell \in \mathcal{J}_f(X_L)$ has the largest index $\ell$. Then $J_\ell$ is the last job in any $s\in S_\ell(X_L)$, i.e., with the same immediate successor $J_{n+1}$. Assume job $J_\ell$ has at least two different immediate predecessors in different $s\in S_\ell(X_L)$ and job $J_k \neq J_0$ is one of them. Since $J_\ell$ is the unique immediate successor of $J_k\in \mathcal{J}_f(X_L)$ according to the lemma condition, every single-machine schedule $s\in\Phi(X_L)$ that contains $J_k$ must also include $J_\ell$. Since (\ref{eq:feasibility}) holds when $j$ is replaced by $k$ and $S_k(X_L)\subseteq S_\ell(X_L)$, we can remove $J_\ell$ from all single-machine schedules $s\in S_\ell(X_L)\backslash S_k(X_L)$ to get a new optimal LP solution $X_L'$ such that $\Phi(X_L')\subseteq \Phi(X_L)$ and $X_L'$ also satisfies the lemma condition and, \emph{additionally}, job $J_\ell \in \mathcal{J}_f(X_L')$ has the same immediate predecessor for every $s\in S_\ell(X_L')$.

Now we repeat, if necessary, the process described in the above paragraph with $X_L$ replaced by $X_L'$, which will eventually lead us to a desired optimal LP solution $X_L^*$.
\end{proof}

\begin{lem}\label{le:same_predecessor_successor}
If each job $J_j\in\mathcal{J}_f(X_L)$ has the same immediate predecessor and successor in every $s\in S_j(X_L)$, then one can find a feasible, and hence optimal, ILP solution.
\end{lem}
\begin{proof}
According to the condition of the lemma, it is not difficult to see that the set $S_j(X_L)$ contains only one single-machine schedule for any $J_j\in \mathcal{J}_f(X_L)$, which implies that the job sets corresponding to different $s\in \Phi(X_L)$ are disjoint each other.

Consider any fixed $s\in \Phi(X_L)$. According to (\ref{eq:feasibility}), $\mathcal{I}({s})$ contains at least $2$ elements due to a fractional value of $x_s^i$ for any $i\in \mathcal{I}(s)$. As a result, the corresponding weights $\{w_i: i\in \mathcal{I}({s})\}$ are all equal since otherwise the objective value could be decreased with a new feasible LP solution. Consequently, by fixing any $i\in \mathcal{I}({s})$, replacing $x_{s}^i$ in $X_L$ with value $1$ and $x_{s}^{i'}$ with value $0$ for any $i'\in \mathcal{I}({s})\backslash\{i\}$, we obtain a new optimal LP solution in which all jobs in $s$ are scheduled non-fractionally.

We repeat the process described in the above paragraph for every $s\in \Phi(X_L)$ to obtain the desirable optimal integral solution thanks to the disjointness of the corresponding job sets as mentioned at the beginning of the proof.
\end{proof}

According to Lemmas~\ref{le:same_successor} and \ref{le:same_predecessor_successor}, uniqueness of immediate successors is a guarantee for the existence of an optimal ILP solution in our search tree. We
said a job $J_j\in \mathcal{J}_f(X_L)$ is \emph{separated} in solution $X_L$ if it has more than one immediate successor in different $s\in S_j(X_L)$. We design our branch-and-price tree such that at each node of the tree, we find a separated job with lowest index, and then create descendant nodes in such a way that each of the descendant nodes corresponds to a fixed immediate successor. Consequently, the branching strategy eventually eliminates all separated jobs, which is a sufficient condition for an integral solution. Such a partitioning strategy allows us to use the same pricing algorithm DP($\mathcal{P}$) as presented in Section~\ref{sct: pricing}, except that the precedence constraints $\mathcal{P}$ are updated along each path of the search tree, gradually fixing the immediate successors.

In searching for the active nodes in our branch-and-price tree, we follow a breadth-first search strategy, i.e., the unexplored neighbor nodes will be searched first before going further to any descendant node. With this strategy, which outperforms the depth-first strategy in our experiments, upper and lower bounds at all the nodes of the same level can be obtained and thus the node having the best lower bound can be searched first.

\bigskip\noindent
\textbf{Branch-and-Price Algorithm}
\begin{description}
\vspace{-6pt}
\item[Step 0.]Set the initial job precedence constraints in (\ref{eq:def-P}) as $\mathcal{P}$ with $\mathcal{J}_j^+=\mathcal{J}\backslash\{J_j\}$ ($j=1,\ldots,n$).
\vspace{-6pt}
\item[Step 1.]Move to the next unexplored node according to the breadth-first rule if there is any unexplored node. Output the best integral solution so far if all nodes have been considered.
    \vspace{-6pt}
\item[Step 2.]At the current node of the branching tree, obtain an initial feasible solution $X$ to the LP problem of the node with Algorithm GR. Starting with $X$, solve the LP problem of the node by column generation. More specifically, unless the reduced cost of each variable in $X$ is already nonnegative, in which case optimality of $X$ is already achieved, we obtain an optimal LP solution $\tilde{X}_{L}$ for the node by applying iteratively pricing algorithm DP($\mathcal{P}$) to add more columns to the LP relaxation problem and solve it again. The objective value of $\tilde{X}_L$ is used as the initial lower bound if the node is the root of the branching tree. If this value is not smaller than the current upper bound, then fathom the node and go to Step~1
\vspace{-6pt}
\item[Step 3.]If either $\tilde{X}_L$ is already integral or, after updating job precedence constraints $\mathcal{P}$ according to the positive components of $\tilde{X}_L$, there is no separated job in $\mathcal{J}_f(\tilde{X}_L)$, in which case the constructive proofs of Lemmas~1 and 2 are applied, we find an optimal integral solution of the node. If the objective value of this integral solution is smaller than the existing upper bound, we update the upper bound with the new value. If the current lower bound is equal to the upper bound, the integral solution is optimal to ILP~(\ref{eq:CG_obj})--(\ref{eq:xs_in0,1}) and we stop. Otherwise, fathom the node and go to Step~1.
\vspace{-6pt}
\item[Step 4.]Find the separate job $J_j$ with the lowest index in the current LP solution $\tilde{X}_{L}$. Branch the current node into $|{\cal J}_j^+|+1$ descendant nodes, one for each immediate successor of job $J_j$ corresponding to the current fractional solution $\tilde{X}_{L}$ and one for all the other possibilities. At each of the first $|{\cal J}_j^+|$ descendant nodes, job $J_j$ is restricted to having a fixed immediate successor from ${\cal J}_j^+$, while at the last descendant node, job $J_j$ is restricted to having an immediate successor not from ${\cal J}_j^+$.
\vspace{-6pt}
\item[Step 5.]For each descendant node, update the overall precedence constraints $\mathcal{P}$ by adding the newly imposed precedence constraint at the node. By taking out from consideration those columns (single-machine schedules) that violate $\mathcal{P}$, we form the LP problem of the node. Go to Step~1.
\end{description}

\section{\label{sec:computation}Computational experiments}

\subsection{On computational effectiveness}\label{se:effectiveness}

In this experiment, we test the computational effectiveness of our branch-and-price (B\&P) algorithm. We use the same parameter settings as in \citet{kroon1997exact} for easy comparison.

\begin{description}
\vspace{-6pt}
\item[Number of jobs:]$n = 100$, $200$ and $300$;
\vspace{-6pt}
\item[Number of job groups:]we consider instances of $g=3$, $4$ and $5$ job groups;
\vspace{-6pt}
\item[Number of machine classes:]all possible machine classes, i.e., $c=2^g-1$, where $g$ denotes the number of job groups;
\vspace{-6pt}
\item[Job duration:]uniformly distributed over $[1, 100]$, $[41, 100]$ and $[81, 100]$;
\vspace{-6pt}
\item[Spread-time:] $L=300$, $400$ and $500$.
\end{description}

Based on the combinations of number of jobs, number of job groups, job duration time and spread time, we have  $3\times3\times3\times3=81$ scenarios in total. For each scenario, we randomly generate 10 instances. All these instances are tested by CPLEX (version 12.3) on a PC Pentium-4 2G with 2Gbytes RAM memory on Linux platform. Time limit is set to one hour, i.e., an instance that cannot be solved within one hour will be marked as unsolved.
We list only the results on the B\&P algorithm because CPLEX cannot solve any of these instances within an hour based on the standard ILP formulation (\ref{eq:ZIP_obj})--(\ref{eq:x_ijk_in0,1}).
For each scenario, we show the average computation time (in seconds) of the B\&P algorithm. The numbers in superscript represent the number of instances that are not solved within an hour. These results are shown in Tables~\ref{n100}--\ref{n300}. We also show in Table~\ref{gap} (with the numbers before and after the slashes) the average relative gap at the root node and that after one hour time limit for the case of $3$ job groups.
\begin{table}[h]
\caption{Running Time (sec) with $n =100$}\label{n100}
\begin{center}{\footnotesize
\begin{tabular}{c|r|r|r|r}
%\hline
Spread $L$   &  Job Duration & 3 Job Groups & 4 Job Groups & 5 Job Groups \\
\hline
\multirow{3}*{300} & $1-100$   & 2.31   & 4.09    & 2.41
\\
  & $41-100$   & 0.43    & 0.68    & 0.44
\\
    & $81-100$   & 0.16    & 0.16    & 0.29
\\
\hline
\multirow{3}*{400} & $1-100$   & 14.29   & 7.72    & 10.84
 \\
  & $41-100$    & 2.41    & 2.06    & 4.32
\\
    & $81-100$    & 0.66    & 0.95    & 2.17
\\
\hline
\multirow{3}*{500} & $1-100$    & 37.12    & 20.13    & 23.63
\\
  & $41-100$    & 13.70    & 6.52    & 12.07
\\
    & $81-100$  & 1.86    & 4.23    & 5.63
\\
\hline
\end{tabular}
}
\end{center}
\end{table}
\begin{table}[h]
\caption{Running Time (sec) with $n = 200$}\label{n200}
\begin{center}
{\footnotesize
\begin{tabular}{c|l|r|r|r}
%\hline
Spread $L$   &  Job Duration  & 3 Job Groups & 4 Job Groups & 5 Job Groups \\
\hline
\multirow{3}*{ 300} & $1-100$    & 26.30    & 37.68   & 41.73
\\
  & $41-100$    & 3.29    & 9.98    & 12.02
\\
    & $81-100$   & 0.65    & 1.14   & 1.98
\\
\hline
\multirow{3}*{ 400} & $1-100$    & 92.49    & 115.04    & 328.58
\\
  & $41-100$    & 273.07   & 98.77    & 24.81
\\
    & $81-100$    & 7.88   & 12.16   & 16.57
\\
\hline
\multirow{3}*{ 500} & $1-100$    & $325.25$    & $595.69$   & 699.14
\\
  & $41-100$    & 193.60    & 170.32   & 204.95
\\
    & $81-100$    & 16.24   & 31.93    & 48.95
\\
\hline
\end{tabular}
}
\end{center}
\end{table}
\begin{table}[h]
\caption{Running Time (sec) with $n = 300$}\label{n300}
\begin{center}
{\footnotesize
\begin{tabular}{c|r|r|r|r}
%\hline
Spread $L$   &  Job Duration  & 3 Job Groups & 4 Job Groups & 5 Job Groups \\
 %&   & CG-Meta & CG & CG-Meta & CG & CG-Meta & CG \\
\hline
\multirow{3}*{ 300} & $1-100$    & $194.05^1$   & 185.01    & 326.10
\\
  & $41-100$    & 36.74    & 31.44  & 53.94
\\
    & $81-100$    & 1.64   & 5.28    & 5.59
\\
\hline
\multirow{3}*{ 400} & $1-100$   & $722.39^2$   & $345.13^4$   & $743.73^3$
\\
  & $41-100$    & $117.49^1$  & 490.93   & 234.99
\\
    & $81-100$   & 107.37  & 69.30   & 101.44
\\
\hline
\multirow{3}*{ 500} & $1-100$   & $2196.64^9$  & $1113.74^9$  & $2070.16^8$
\\
  & $41-100$   & $1501.06^1$& $1244.23^2$  & $1485.32^5$
\\
    & $81-100$    & 217.51   & 609.01   & 643.59
\\
\hline
\end{tabular}
}
\end{center}
\end{table}

\begin{table}[h]
\caption{Average Relative Gap (\%) with 3 Job Groups}\label{gap}
\begin{center}
{\footnotesize
\begin{tabular}{c|r|r|r|r}
%\hline
Spread $L$   &  Job Duration  & $n=100$ & $n=200$ & $n=300$ \\
\hline
\multirow{3}*{ 300} & $1-100$    & 14.25~/~0.00   & 16.02~/~0.00    & 12.69~/~1.50
\\
  & $41-100$    & 8.48~/~0.00    & 11.39~/~0.00  & 10.67~/~0.00
\\
    & $81-100$    & 6.02~/~0.00   & 9.36~/~0.00    & 10.74~/~0.00
\\
\hline
\multirow{3}*{ 400} & $1-100$   & 14.74~/~0.00   & 14.65~/~0.00  & 11.06~/~5.15
\\
  & $41-100$    & 10.29~/~0.00  & 11.56~/~0.00   & 6.94~/~1.08
\\
    & $81-100$   & 11.75~/~0.00  & 10.43~/~0.00   & 12.42~/~0.00
\\
\hline
\multirow{3}*{ 500} & $1-100$   & 17.34~/~0.00  & 14.83~/~0.00  & 12.01~/~10.43
\\
  & $41-100$   & 13.14~/~0.00 & 10.52~/~0.00  & 10.66~/~4.11
\\
    & $81-100$    & 10.88~/~0.00   & 8.72~/~0.00   & 9.37~/~0.00
\\
\hline
\end{tabular}
}
\end{center}
\end{table}

From Tables~\ref{n100}--\ref{n300}, we find that the B\&P algorithm can solve all the instances with 100 and 200 jobs in a reasonable amount of time, but the instances with 300 jobs are sometimes difficult, especially in the case with large spread time. In the appendix, we provide a neighborhood search heuristic to be embedded into our B\&P algorithm and show that the embedded B\&P algorithm runs faster in many difficult scenarios. This heuristic can generate local optimal integer solutions by solving a reduced mixed integer program after randomly removing a number of columns. These local optimal integer solutions help improve the lower bounds and hence speed up the embedded B\&P algorithm.

In Tables~\ref{n100}--\ref{n300}, the running time for solving each case increases with the length of spread-time and also increases with the range of job duration, which accord with our intuition.

\begin{table}[t]
\caption{Total Flexibility with Weights $(1,\ldots,1)$}\label{n111}
\begin{center}{\footnotesize
\begin{tabular}{c|r|rrrrrrrrrr|r}
%\hline
Spread  &  Job  & \multicolumn{10}{c|}{Machine Distributions } & Avg\# \\
$L$ & Duration &\multicolumn{10}{c|}{ in Percentage}  &    \\
\hline
\multirow{3}*{300} & $1-100$    & 15.7 & 11.0 & 3.3 & 7.9 & 24.3 & 19.1 & 10.8 & 7.2 & 0.6 & 0.0 & 16.4 \\
  & $41-100$    & 30.0 & 8.5 & 2.9 & 7.1 & 28.4 & 15.8 & 4.6 & 2.1 & 0.6 & 0.0 & 19.7 \\
    & $81-100$    & 31.0 & 37.4 & 2.7 & 4.9 & 17.0 & 5.7 & 0.9 & 0.5 & 0.0 & 0.0 & 23.0 \\
\hline
\multirow{3}*{400} & $1-100$   & 9.9 & 8.0 & 2.4 & 2.3 & 22.4 & 19.7 & 20.8 & 10.1 & 4.5 & 0.0 & 13.6 \\
  & $41-100$    & 16.2 & 11.0 & 2.1 & 9.6 & 24.7 & 21.6 & 9.7 & 3.8 & 0.7 & 0.6 & 16.3 \\
    & $81-100$   & 21.1 & 9.2 & 1.6 & 9.4 & 31.7 & 14.9 & 6.5 & 5.4 & 0.0 & 0.0 & 18.4 \\
\hline
\multirow{3}*{500} & $1-100$   & 12.1 & 3.2 & 1.0 & 3.3 & 14.8 & 20.8 & 19.9 & 10.4 & 11.2 & 3.3 & 11.8 \\
  & $41-100$   & 16.6 & 4.5 & 1.3 & 7.9 & 19.1 & 23.0 & 15.9 & 8.2 & 3.7 & 0.0  & 14.2 \\
    & $81-100$  &  18.1 & 5.7 & 0.6 & 9.8 & 21.9 & 26.4 & 9.5 & 6.8 & 1.3 & 0.0  & 15.8 \\
\hline
\end{tabular}}
\end{center}
\end{table}

\subsection{On machine flexibility}

Flexibility is one of the most important aspects of production systems. \citet{jordan1995principles} study the process flexibility resulted from being able to build different types of products at the same time in the same manufacturing plant or on the same production line. They show that limited process flexibility may still yield most of the benefits of total process flexibility. In this experiment, we test this statement on machine flexibility for processing fixed interval jobs, where the flexibility of a machine is measured by its capability of processing jobs of different classes.

We use the same parameter settings as in the experiment of Section 4.1, except that the number of job groups is set to 10 and the number of jobs is set to 50. The total number of machine classes is $c=2^{10}-1$ with 10 \emph{tiers} of classes, where each of $C_i^{10}=10!/(i!(10-i)!)$ tier-$i$ classes consists of those machines that are each capable of processing jobs of $i$ job groups. Denote by $\tilde{M}^i$ the set of machines of tier-$i$ classes ($i=1, \ldots, 10$). Therefore, the higher the value of $i$, the more capable the machines of $\tilde{M}^i$.

Firstly, using the average number of used machines (Avg\#) as an indicator, we compare two extreme cases. In one case, we set weights (costs) $w_i=1$ for all 10 machine sets $\tilde{M}^i$ ($i=1,\ldots,10$), while in the other case, we set the weights to 1 for  $\tilde{M}^1$ and the weights to 999 for sets $\tilde{M}^i$ ($i\ge 2$). Obviously, we have total flexibility in the former case and  no flexibility in the latter case. We randomly generate 10 instances for each scenario of the two cases and the results are shown in Tables~\ref{n111} and \ref{n1999}, where we also provide percentage distribution of used machines across the 10 machine sets.
\begin{table}[htbp]
\caption{No Flexibility with Weights $(1,999,\ldots,999)$}\label{n1999}
\begin{center}{\footnotesize
\begin{tabular}{c|r|rrrrrrrrrr|r}
%\hline
Spread  &  Job & \multicolumn{10}{c|}{Machine Distributions } & Avg\# \\
$L$ & Duration &\multicolumn{10}{c|}{ in Percentage}  &    \\
\hline
\multirow{3}*{300} & $1-100$ &100.0 & 0.0 & 0.0 & 0.0 & 0.0 & 0.0 & 0.0 & 0.0 & 0.0 & 0.0 & 31.7 \\
& $41-100$ &100.0 & 0.0 & 0.0 & 0.0 & 0.0 & 0.0 & 0.0 & 0.0 & 0.0 & 0.0 & 34.0 \\
& $81-100$ &100.0 & 0.0 & 0.0 & 0.0 & 0.0 & 0.0 & 0.0 & 0.0 & 0.0 & 0.0 & 36.5 \\
 \hline
\multirow{3}*{400} & $1-100$ &100.0 & 0.0 & 0.0 & 0.0 & 0.0 & 0.0 & 0.0 & 0.0 & 0.0 & 0.0 & 27.8 \\
& $41-100$ &100.0 & 0.0 & 0.0 & 0.0 & 0.0 & 0.0 & 0.0 & 0.0 & 0.0 & 0.0 & 30.4 \\
& $81-100$ &100.0 & 0.0 & 0.0 & 0.0 & 0.0 & 0.0 & 0.0 & 0.0 & 0.0 & 0.0 & 31.4 \\
 \hline
\multirow{3}*{500} & $1-100$ &100.0 & 0.0 & 0.0 & 0.0 & 0.0 & 0.0 & 0.0 & 0.0 & 0.0 & 0.0 & 24.1 \\
& $41-100$ &100.0 & 0.0 & 0.0 & 0.0 & 0.0 & 0.0 & 0.0 & 0.0 & 0.0 & 0.0 & 27.2 \\
& $81-100$ &100.0 & 0.0 & 0.0 & 0.0 & 0.0 & 0.0 & 0.0 & 0.0 & 0.0 & 0.0 & 29.5 \\
 \hline
\end{tabular}}
\end{center}
\end{table}

We see from Tables~\ref{n111} and \ref{n1999} that, when compared with the case of no flexibility, the Avg\# values in the case of total flexibility are significantly lower. When there is total flexibility, the selected machines can be freely from all 10 machine sets. In contrast, the selected machines are all from set $\tilde{M}^1$ (100\%) if there is no flexibility.

Now let us look at the relationship between machine flexibility and Avg\# values. By changing the weights of the machine sets, we gradually increase machine flexibility. The experiment results are shown in Tables~\ref{n11999} and \ref{n111999}.

\begin{table}[h]
\caption{Little Flexibility with Weights $(1,1,999,\ldots,999)$ \label{n11999}}
\begin{center}
{\footnotesize
\begin{tabular}{c|r|rrrrrrrrrr|r}
%\hline
Spread  &  Job   & \multicolumn{10}{c|}{Machine Distributions } & Avg\# \\
$L$ & Duration &\multicolumn{10}{c|}{ in Percentage}  &    \\
\hline
\multirow{3}*{300} & $1-100$ &6.0 & 94.0 & 0.0 & 0.0 & 0.0 & 0.0 & 0.0 & 0.0 & 0.0 & 0.0 & 17.9 \\
& $41-100$ &9.6 & 90.4 & 0.0 & 0.0 & 0.0 & 0.0 & 0.0 & 0.0 & 0.0 & 0.0 & 20.6 \\
& $81-100$ &27.6 & 72.4 & 0.0 & 0.0 & 0.0 & 0.0 & 0.0 & 0.0 & 0.0 & 0.0 & 23.8 \\
 \hline
\multirow{3}*{400} & $1-100$ &2.0 & 98.0 & 0.0 & 0.0 & 0.0 & 0.0 & 0.0 & 0.0 & 0.0 & 0.0 & 15.1 \\
& $41-100$ &5.0 & 95.0 & 0.0 & 0.0 & 0.0 & 0.0 & 0.0 & 0.0 & 0.0 & 0.0 & 17.0 \\
& $81-100$ &9.1 & 90.9 & 0.0 & 0.0 & 0.0 & 0.0 & 0.0 & 0.0 & 0.0 & 0.0 & 18.6 \\
 \hline
\multirow{3}*{500} & $1-100$ &0.8 & 99.2 & 0.0 & 0.0 & 0.0 & 0.0 & 0.0 & 0.0 & 0.0 & 0.0 & 13.2 \\
& $41-100$ &1.8 & 98.2 & 0.0 & 0.0 & 0.0 & 0.0 & 0.0 & 0.0 & 0.0 & 0.0 & 15.2 \\
& $81-100$ &5.3 & 94.7 & 0.0 & 0.0 & 0.0 & 0.0 & 0.0 & 0.0 & 0.0 & 0.0 & 16.6 \\
 \hline
\end{tabular}
}
\end{center}
\end{table}
\begin{table}[h]
\caption{Small Flexibility with Weights $(1,1,1,999,\ldots,999)$ \label{n111999}}
\begin{center}
{\footnotesize
\begin{tabular}{c|r|rrrrrrrrrr|r}
%\hline
Spread  &  Job  & \multicolumn{10}{c|}{Machine Distributions } & Avg\# \\
$L$ & Duration &\multicolumn{10}{c|}{ in Percentage}  &    \\
\hline
\multirow{3}*{300} & $1-100$ &12.8 & 7.8 & 79.4 & 0.0 & 0.0 & 0.0 & 0.0 & 0.0 & 0.0 & 0.0 & 16.5 \\
& $41-100$ &24.3 & 23.6 & 52.0 & 0.0 & 0.0 & 0.0 & 0.0 & 0.0 & 0.0 & 0.0 & 19.7 \\
& $81-100$ &31.3 & 38.3 & 30.4 & 0.0 & 0.0 & 0.0 & 0.0 & 0.0 & 0.0 & 0.0 & 23.0 \\
 \hline
\multirow{3}*{400} & $1-100$ &4.0 & 4.8 & 91.2 & 0.0 & 0.0 & 0.0 & 0.0 & 0.0 & 0.0 & 0.0 & 13.6 \\
& $41-100$ &14.4 & 6.0 & 79.6 & 0.0 & 0.0 & 0.0 & 0.0 & 0.0 & 0.0 & 0.0 & 16.3 \\
& $81-100$ &16.4 & 11.5 & 72.2 & 0.0 & 0.0 & 0.0 & 0.0 & 0.0 & 0.0 & 0.0 & 18.4 \\
 \hline
\multirow{3}*{500} & $1-100$ &4.7 & 0.8 & 94.4 & 0.0 & 0.0 & 0.0 & 0.0 & 0.0 & 0.0 & 0.0 & 11.8 \\
& $41-100$ &11.7 & 5.4 & 83.0 & 0.0 & 0.0 & 0.0 & 0.0 & 0.0 & 0.0 & 0.0 & 14.2 \\
& $81-100$ &11.2 & 13.8 & 75.0 & 0.0 & 0.0 & 0.0 & 0.0 & 0.0 & 0.0 & 0.0 & 15.8 \\
 \hline
\end{tabular}
}
\end{center}
\end{table}

From Tables~\ref{n11999} and \ref{n111999}, we find that limited flexibility can achieve almost the same benefits of total flexibility. As shown in Table~\ref{n11999}, the Avg\# values are close to the values in Table~\ref{n111} and those in Table~\ref{n111999} are almost the same as the values in Table~\ref{n111}. Our calculations show that, when machine choices are limited to sets $\tilde{M}^1$ and $\tilde{M}^2$, there is a 92.3\% reduction in average Avg\# values and, when machine choices are limited to $\tilde{M}^1$, $\tilde{M}^2$ and $\tilde{M}^3$, the reduction achieves 99.9\%! Therefore, higher machine flexibility is almost useless in reducing Avg\# values.

Finally, we examine the impact of more general weights on machine selection. We apply a percentage weight increment $p$ when a machine is made capable of processing one more group of jobs. We set the weights to 1 for machines that can process jobs from one group only and the weights to $1+kp$ for machines which can process $k$ more groups of jobs. Let $p=0.03$, $0.20$, $0.80$ and $1.00$, respectively. We present our test results in Tables~\ref{n1.03}--\ref{n1234}.

\begin{table}[h!]
\caption{Machine Selection with Weights $1+kp$ ($k=0,1,\ldots,9$)
and $p=0.03$}\label{n1.03}
\begin{center}
{\footnotesize
\begin{tabular}{c|r|rrrrrrrrrr|r}
%\hline
Spread  &  Job   & \multicolumn{10}{c|}{Machine Distributions } & Avg\# \\
$L$ & Duration &\multicolumn{10}{c|}{ in Percentage}  &    \\
\hline
\multirow{3}*{300} & $1-100$ &25.6 & 47.5 & 23.7 & 3.1 & 0.0 & 0.0
 & 0.0 & 0.0 & 0.0 & 0.0 & 16.4 \\
& $41-100$ &39.3 & 38.7 & 21.1 & 0.9 & 0.0 & 0.0 & 0.0 & 0.0 & 0.0 & 0.0 & 19.7
\\
& $81-100$ &42.8 & 42.2 & 15.0 & 0.0 & 0.0 & 0.0 & 0.0 & 0.0 & 0.0 & 0.0 & 23.0
\\
 \hline
\multirow{3}*{400} & $1-100$ &21.4 & 44.3 & 31.2 & 3.1 & 0.0 & 0.0
 & 0.0 & 0.0 & 0.0 & 0.0 & 13.6 \\
& $41-100$ &30.4 & 44.0 & 24.4 & 1.3 & 0.0 & 0.0 & 0.0 & 0.0 & 0.0 & 0.0 & 16.3 \\
& $81-100$ &37.0 & 47.0 & 15.5 & 0.5 & 0.0 & 0.0 & 0.0 & 0.0 & 0.0 & 0.0 & 18.4
\\
 \hline
\multirow{3}*{500} & $1-100$ &22.4 & 44.0 & 30.9 & 2.7 & 0.0 & 0.0
 & 0.0 & 0.0 & 0.0 & 0.0 & 11.8 \\
& $41-100$ &27.6 & 45.0 & 25.4 & 1.9 & 0.0 & 0.0 & 0.0 & 0.0 & 0.0 & 0.0 & 14.2
\\
& $81-100$ &27.7 & 50.7 & 19.6 & 2.1 & 0.0 & 0.0 & 0.0 & 0.0 & 0.0 & 0.0 & 15.8
\\
 \hline
\end{tabular}
}
\end{center}
\end{table}
\begin{table}[h!]
\caption{Machine Selection with Weights $1+kp$ ($k=0,1,\ldots,9$)
and $p=0.20$}\label{n1.20}
\begin{center}
{\footnotesize
\begin{tabular}{c|r|rrrrrrrrrr|r}
%\hline
Spread  &  Job  & \multicolumn{10}{c|}{Machine Distributions } & Avg\# \\
$L$ & Duration &\multicolumn{10}{c|}{ in Percentage}  &    \\
\hline
\multirow{3}*{300} & $1-100$ &25.8 & 48.7 & 21.1 & 3.7 & 0.7 & 0.0
 & 0.0 & 0.0 & 0.0 & 0.0 & 16.4 \\
& $41-100$ &36.9 & 43.5 & 18.5 & 1.1 & 0.0 & 0.0 & 0.0 & 0.0 & 0.0 & 0.0 & 19.7
\\
& $81-100$ &44.1 & 40.5 & 15.4 & 0.0 & 0.0 & 0.0 & 0.0 & 0.0 & 0.0 & 0.0 & 23.0
\\
 \hline
\multirow{3}*{400} & $1-100$ &22.1 & 43.5 & 31.4 & 2.2 & 0.8 & 0.0
 & 0.0 & 0.0 & 0.0 & 0.0 & 13.6 \\
& $41-100$ &26.8 & 51.7 & 20.3 & 0.7 & 0.6 & 0.0 & 0.0 & 0.0 & 0.0 & 0.0 & 16.3
\\
& $81-100$ &38.6 & 43.3 & 18.1 & 0.0 & 0.0 & 0.0 & 0.0 & 0.0 & 0.0 & 0.0 & 18.4
\\
 \hline
\multirow{3}*{500} & $1-100$ &24.9 & 39.8 & 30.8 & 4.5 & 0.0 & 0.0
 & 0.0 & 0.0 & 0.0 & 0.0 & 11.8 \\
& $41-100$ &25.5 & 51.8 & 20.5 & 1.6 & 0.6 & 0.0 & 0.0 & 0.0 & 0.0 & 0.0 & 14.3
\\
& $81-100$ &29.1 & 48.7 & 20.8 & 0.7 & 0.7 & 0.0 & 0.0 & 0.0 & 0.0 & 0.0 & 15.8
\\
 \hline
\end{tabular}
}
\end{center}
\end{table}
\begin{table}[h!]
\caption{Machine Selection with Weights $1+kp$ ($k=0,1,\ldots,9$) and
$p=0.80$}\label{n1.80}
\begin{center}
{\footnotesize
\begin{tabular}{c|r|rrrrrrrrrr|r}
%\hline
Spread  &  Job   & \multicolumn{10}{c|}{Machine Distributions } & Avg\# \\
$L$ & Duration &\multicolumn{10}{c|}{ in Percentage}  &    \\
\hline
\multirow{3}*{300} & $1-100$ &36.8 & 45.6 & 16.3 & 0.6 & 0.7 & 0.0 & 0.0 & 0.0 & 0.0 & 0.0 & 17.5 \\
& $41-100$ &53.3 & 34.4 & 10.8 & 1.5 & 0.0 & 0.0 & 0.0 & 0.0 & 0.0 & 0.0 & 21.3 \\
& $81-100$ &63.2 & 27.0 & 9.8 & 0.0 & 0.0 & 0.0 & 0.0 & 0.0 & 0.0 & 0.0 & 25.0 \\
 \hline
\multirow{3}*{400} & $1-100$ &36.7 & 40.1 & 21.9 & 0.7 & 0.6 & 0.0 & 0.0 & 0.0 & 0.0 & 0.0 & 14.9 \\
& $41-100$ &40.8 & 43.5 & 15.1 & 0.6 & 0.0 & 0.0 & 0.0 & 0.0 & 0.0 & 0.0 & 17.4 \\
& $81-100$ &51.6 & 38.7 & 8.7 & 1.0 & 0.0 & 0.0 & 0.0 & 0.0 & 0.0 & 0.0 & 19.8 \\
 \hline
\multirow{3}*{500} & $1-100$ &31.0 & 49.9 & 16.6 & 2.5 & 0.0 & 0.0 & 0.0 & 0.0 & 0.0 & 0.0 & 12.7 \\
& $41-100$ &40.9 & 43.2 & 14.9 & 1.1 & 0.0 & 0.0 & 0.0 & 0.0 & 0.0 & 0.0 & 15.5 \\
& $81-100$ &37.7 & 46.9 & 14.2 & 1.2 & 0.0 & 0.0 & 0.0 & 0.0 & 0.0 & 0.0 & 16.6 \\
 \hline
\end{tabular}
}
\end{center}
\end{table}
\begin{table}[h!]
\caption{Machine Selection with Weights $1+kp$ ($k=0,1,\ldots,9$) and $p=1.00$}\label{n1234}
\begin{center}
{\footnotesize
\begin{tabular}{c|r|rrrrrrrrrr|r}
%\hline
Spread  &  Job   & \multicolumn{10}{c|}{Machine Distributions } & Avg\# \\
$L$ & Duration &\multicolumn{10}{c|}{ in Percentage}  &    \\
\hline
\multirow{3}*{300} & $1-100$    & 100.0 & 0.0 & 0.0 & 0.0 & 0.0 & 0.0 & 0.0 & 0.0 & 0.0 & 0.0  & 31.7 \\
  & $41-100$    & 100.0 & 0.0 & 0.0 & 0.0 & 0.0 & 0.0 & 0.0 & 0.0 & 0.0 & 0.0 & 34.0 \\
    & $81-100$    & 100.0 & 0.0 & 0.0 & 0.0 & 0.0 & 0.0 & 0.0 & 0.0 & 0.0 & 0.0 & 36.5 \\
\hline
\multirow{3}*{400} & $1-100$   & 100.0 & 0.0 & 0.0 & 0.0 & 0.0 & 0.0 & 0.0 & 0.0 & 0.0 & 0.0 & 27.8 \\
  & $41-100$    & 99.7 & 0.3 & 0.0 & 0.0 & 0.0 & 0.0 & 0.0 & 0.0 & 0.0 & 0.0 & 30.3 \\
    & $81-100$   & 100.0 & 0.0 & 0.0 & 0.0 & 0.0 & 0.0 & 0.0 & 0.0 & 0.0 & 0.0 & 31.4 \\
\hline
\multirow{3}*{500} & $1-100$   & 99.6 & 0.4 & 0.0 & 0.0 & 0.0 & 0.0 & 0.0 & 0.0 & 0.0 & 0.0  & 24.0 \\
  & $41-100$   & 100.0 & 0.0 & 0.0 & 0.0 & 0.0 & 0.0 & 0.0 & 0.0 & 0.0 & 0.0  & 27.2 \\
    & $81-100$  &  99.6 & 0.4 & 0.0 & 0.0 & 0.0 & 0.0 & 0.0 & 0.0 & 0.0 & 0.0  & 29.4 \\
\hline
\end{tabular}
}
\end{center}
\end{table}

From Tables~\ref{n1.03}--\ref{n1234} we find that limited machine flexibility by using machines that can process at most three or four groups of jobs are the best choices in most situations ($p= 0.03$, $0.20$ and $0.80$). When multiple-purpose machines are really expensive ($p=1.00$), machine flexibility offers no net benefit.

\section{Conclusions}\label{sec:conclusions}

We have addressed the tactical fixed job scheduling problem with spread-time constraints (TFJSS) and presented a  branch-and-price algorithm to solve randomly generated instances with up to 300 jobs within one hour. Compared with CPLEX, our algorithm is dominantly more efficient. We have also investigated the impact of machine flexibility by computational experiments and found that limited machine flexibility, by using machines that are capable of processing at most three or four groups of jobs, is the best choice in most situations. We additionally afford a neighborhood search heuristic in the appendix to be embedded to improve our branch-and-price algorithm for some difficult instances of the TFJSS problem. We remark that further exploration can be made to possibly improve this embedding process by, for example, a dynamic setting of the two parameters (see the appendix for details), $\theta$ and the maximum number of iterations, instead of the current fixed values.

\section*{Acknowledgements}

\noindent Xiandong Zhang thanks the support of National Natural Science Foundation of China (Projects No. 71171058, No. 70832002 and No. 70971100).

%\vspace{30pt}

\bibliographystyle{model1-num-names}
\bibliography{TFJSS}

%\vspace{50pt}

\appendix

\section{An embedding neighborhood search heuristic}

In this part, we present a neighborhood search heuristic, abbreviated as NS, that can be embedded into the B\&P algorithm for an improved upper bound. The main idea of the NS heuristic is to find randomly some local optimal integer solutions. The best feasible solution is then saved as the current upper bound in the B\&P algorithm.

In the column generation step of the B\&P algorithm, columns are generated and added to the constraint matrix of the LP relaxation under consideration. The NS heuristic generates a local optimal integer solution through a two-step procedure based on the current columns.  The idea is similar to those for the set covering problem studied in \citet{caprara1999heuristic}, \citet{caprara2000algorithms} and \citet{lan2007effective}.

First, a number of current columns are removed randomly from the current LP solution $X$, after which the solution may become infeasible due to some uncovered rows. Then the partial solution is made feasible again by solving a reduced mixed integer program that is made up of the uncovered rows and the columns covering these rows. Since CPLEX is used here at each iteration, the size of the program needs to be controlled.  We use parameter ($0<\theta<1$) to control how many columns will be removed randomly, which is equal to $\theta |X|$, where $|X|$ denotes the number of columns used in the current LP solution.

Once a better solution is found, the NS will be executed to find the next better solution, if any, in the neighborhood of the best solution so far. This is done iteratively until the number of iterations reaches a given number.
After each iteration, if a better ILP solution is identified, all redundant columns, those corresponding to positive components of $X$ removal of which still keeps $X$ feasible, will be removed with the largest cost removed first.

Here we used the same problems as in Section~\ref{se:effectiveness} to test the effectiveness of our embedded B\&P algorithm, which we denote by B\&P-NS. In this experiment, we set $\theta$ at $0.6$ and the number of iterations at $40$ by experience. We present our test results in Table~\ref{n300n}, from which we see that B\&P-NS outperforms B\&P for all difficult instances and the total number of unsolved instances is reduced by 41\%.

\begin{table}[h!]
\caption{Comparison on Running Time (sec)} \label{n300n}
\begin{center}
{\footnotesize
\begin{tabular}{c|r|rr|rr|rr}
%\hline
   & Job Duration & \multicolumn{2}{|c}{3 Job Groups}
          & \multicolumn{2}{|c}{4 Job Groups} & \multicolumn{2}{|c}{5 Job Groups} \\
  & Time  & B\&P-NS & B\&P & B\&P-NS & B\&P & B\&P-NS & B\&P \\
\hline
\multirow{2}*{$n=200$} & $1-100$  & 316.44  & $325.25$  & 344.78  & $595.69$  & 286.30  & 699.14
\\
  & $41-100$  & 62.80  & 193.60  & 135.31  & 170.32  & 89.85  & 204.95
\\
$L=500$    & $81-100$  & 15.25  & 16.24  & 19.57  & 31.93  & 28.17  & 48.95
\\
\hline
\multirow{2}*{$n=300$} & $1-100$  & 234.09  & $722.39^2$ & 190.64  & $345.13^4$ & 297.76  & $743.73^3$
\\
  & $41-100$  & 152.72  & $117.488^1$ & 85.51  & 490.93  & 64.47  & 234.99
\\
$L=400$ & $81-100$  & 86.40  & 107.37 & 44.82  & 69.30  & 67.52  & 101.44
\\
\hline
\multirow{2}*{$n=300$} & $1-100$  & $1926.02^8$ & $2196.64^9$ & $1858.9^5$ & $1113.74^9$ & $1424.84^5$ & $2070.16^8$
\\
  & $41-100$  & $674.91^2$ & $1501.06^1$ & $742.07^2$ & $1244.23^2$ & $690.38^4$ & $1485.32^5$
\\
$L=500$    & $81-100$  & 151.75  & 217.51  & 117.73  & 609.01  & 172.22  & 643.59
\\
\hline
\end{tabular}
}
\end{center}
\end{table}

%%\label{}



%% References
%% Usage of \citet is as follows:







% \end{thebibliography}

\end{document}